\documentclass[twocolumn,10pt]{IEEEtran}

\usepackage[acronym]{glossaries}
\usepackage{amsmath,amssymb,amsthm}
\usepackage{graphicx}
\usepackage{comment}
\usepackage{xcolor}
\usepackage{booktabs}
\usepackage{subfigure}

\newacronym{awgn}{AWGN}{additive white Gaussian noise}
\newacronym{miso}{MISO}{multiple-input single-output}
\newacronym{rf}{RF}{radio frequency}
\newacronym{sinr}{SINR}{signal-to-interference-plus-noise ratio}
\newacronym{pcb}{PCB}{printed circuit board}
\newacronym{iid}{i.i.d.}{independent and identically distributed}

\newtheorem{definition}{Definition}

\newtheorem{proposition}{Proposition}

\begin{document}
\bstctlcite{BSTcontrol}

\title{\huge Low-Complexity Planar Beyond-Diagonal RIS\\Architecture Design Using Graph Theory}

\author{Matteo~Nerini,~\IEEEmembership{Member,~IEEE},
        Zheyu~Wu,\\
        Shanpu~Shen,~\IEEEmembership{Senior Member,~IEEE},
        Bruno~Clerckx,~\IEEEmembership{Fellow,~IEEE}
        
\thanks{M. Nerini, Z. Wu, and B. Clerckx are with the Department of Electrical and Electronic Engineering, Imperial College London, London SW7 2AZ, U.K. (e-mail: \{m.nerini20, zheyu.wu, b.clerckx\}@imperial.ac.uk).}
\thanks{S. Shen is with the State Key Laboratory of Internet of Things for Smart City and Department of Electrical and Computer Engineering, University of Macau, Macau, China. (e-mail: shanpushen@um.edu.mo)}}
\maketitle

\begin{abstract}
Reconfigurable intelligent surfaces (RISs) enable programmable control of the wireless propagation environment and are key enablers for future networks.
Beyond-diagonal RIS (BD-RIS) architectures enhance conventional RIS by interconnecting elements through tunable impedance components, offering greater flexibility with higher circuit complexity.
However, excessive interconnections between BD-RIS elements require multi-layer \gls{pcb} designs, increasing fabrication difficulty.
In this letter, we use graph theory to characterize the BD-RIS architectures that can be realized on double-layer \glspl{pcb}, denoted as planar-connected RISs.
Among the possible planar-connected RISs, we identify the ones with the most degrees of freedom, expected to achieve the best performance under practical constraints.
\end{abstract}

\glsresetall

\begin{IEEEkeywords}
Beyond-diagonal reconfigurable intelligent surface (BD-RIS), graph theory, planar.
\end{IEEEkeywords}

\section{Introduction}

Reconfigurable intelligent surface (RIS) is a technology designed to control wireless communication environments \cite{wu19,dir20,wu21}.
An RIS is a surface consisting of multiple elements with reconfigurable reflecting properties that can be deployed in the propagation environment and used to shape signal propagation.
By steering the incident signal towards the intended direction and suppressing interference, RIS offers a promising green solution for next-generation communication networks such as 6G.
An RIS has been commonly implemented by connecting each RIS element to a tunable impedance, and individually controlling the reflection coefficient of each element.

The conventional RIS architecture, referred to as single-connected RIS \cite{she22}, has been generalized by interconnecting the RIS elements to each other through tunable impedance components in \cite{she22}, leading to the general family of beyond-diagonal RIS (BD-RIS).
When all RIS elements are interconnected to each other, the resulting architecture is denoted as fully-connected RIS \cite{she22}.
To trade flexibility and circuit complexity, the group-connected RIS was also proposed, where the RIS elements are divided into groups of equal size and are interconnected with each other if and only if within the same group \cite{she22}.
The presence of interconnections enables BD-RIS to have additional capabilities than conventional RIS, such as the possibility to achieve full-space coverage \cite{li23,li25}.

A fundamental question in BD-RIS is whether the additional circuit complexity justifies the performance improvement.
Therefore, it is crucial to design BD-RIS architectures that have low circuit complexity and achieve high performance gains at the same time.
This can be done by modeling BD-RIS architectures as graphs, and using graph theory to gain insights into their circuit complexity and performance \cite{ner24}.
Following this graph theoretical modeling, two BD-RIS architectures have been proposed in \cite{ner24}, i.e., the forest- and tree-connected RISs, that are proven to be the least complex architectures achieving the same performance of group- and fully-connected RISs in single-user systems.
With a focus on single-user systems, the fundamental limits of the trade-off between circuit complexity and performance achievable with BD-RIS have been derived in closed-form in \cite{ner23}, and also extended to dual-polarized systems in \cite{ner25}.
Further literature has also extended the results of \cite{ner24} to multi-user systems, proposing stem-connected RISs \cite{zho25a,zho25b} and band-connected RISs \cite{wu25} as BD-RISs achieving the same performance as fully-connected RISs with much reduced circuit complexity.

Previous research on the design of BD-RIS architectures aimed at minimizing the circuit complexity of the BD-RIS measured by the number of tunable impedance components \cite{ner24}-\cite{wu25}.
However, when the number of interconnections is too large, they inevitably cross each other, requiring a multi-layer \gls{pcb} design with vias to realize those crossings.
This increases the implementation cost, time, and probability of fabrication failure, posing a practical challenge.
A significantly more convenient approach is to use a double-layer \gls{pcb}, where one layer has the BD-RIS interconnections and the other serves as the ground plane.
In this work, we identify the requirements that a BD-RIS architecture must satisfy to be realizable on a double-layer \gls{pcb}, in which one layer is entirely reserved for the ground plane.
Furthermore, we characterize the BD-RIS architectures that include the largest number of tunable components while remaining compatible with double-layer \gls{pcb} implementation.

Specifically, the contributions of this letter are as follows.
\textit{First}, we use graph theory to characterize the BD-RIS architectures that can be implemented in a double-layer \gls{pcb}, which we refer to as planar-connected RISs.
\textit{Second}, we examine existing BD-RIS architectures to determine whether they are planar-connected.
\textit{Third}, we characterize the planar-connected RISs with the most degrees of freedom, denoted as maximal-planar-connected RISs, which are expected to achieve the best performance with the constraint of double-layer \gls{pcb}.

\section{BD-RIS Graph Theoretical Model}

Consider an $N$-element BD-RIS, where the RIS elements are connected to an $N$-port reconfigurable microwave network, commonly represented through its scattering matrix $\boldsymbol{\Theta}\in\mathbb{C}^{N\times N}$ \cite{she22}.
According to multiport network theory \cite[Chapter~4]{poz11}, $\boldsymbol{\Theta}$ can be expressed as a function of the admittance matrix of the BD-RIS $\mathbf{Y}\in\mathbb{C}^{N\times N}$ as
\begin{equation}
\boldsymbol{\Theta}=\left(\mathbf{I}+Z_{0}\mathbf{Y}\right)^{-1}\left(\mathbf{I}-Z_{0}\mathbf{Y}\right),\label{eq:T(Y)}
\end{equation}
where $Z_0$ is the reference impedance, set to $Z_0=50\;\Omega$.
In this work, we consider the admittance matrix representation since it can be directly linked to the tunable impedance (or admittance) components implementing the BD-RIS, assumed to be reciprocal.
Specifically, denoting as $Y_{n}$ the tunable admittance connecting the $n$th RIS element to ground and as $Y_{n,m}=Y_{m,n}$ the tunable admittance interconnecting the $n$th and the $m$th RIS elements, we have
\begin{equation}
\left[\mathbf{Y}\right]_{n,m}=\begin{cases}
-Y_{n,m} & n\neq m\\
Y_{n}+\sum_{k\neq n}Y_{n,k} & n=m
\end{cases},\label{eq:Yij}
\end{equation}
for $n,m=1,\ldots,N$.
From \eqref{eq:Yij}, we observe that if the $n$th and $m$th RIS elements are not interconnected, i.e., $Y_{n,m}=0$, the $(n,m)$th entry of the admittance matrix $\mathbf{Y}$ is zero, i.e., $[\mathbf{Y}]_{n,m}=0$.

To model the general circuit topology of a BD-RIS, we resort to graph theory by briefly recalling the model developed in \cite{ner24}.
Following this model, the circuit topology of any BD-RIS is represented through a graph $\mathcal{G}=(\mathcal{V},\mathcal{E})$, where $\mathcal{V}$ and $\mathcal{E}$ are the vertex set and the edge set of $\mathcal{G}$, respectively.
The vertex set is given by the indexes of the RIS elements, i.e. $\mathcal{V}=\{v_1,v_2,\ldots,v_N\}$, and the edge set is given by 
\begin{equation}
\mathcal{E}=\left\{ \left(v_n,v_m\right)|\:v_n,v_m\in\mathcal{V},\:Y_{n,m}\neq0,\:n\neq m\right\},
\end{equation}
which means that vertices $v_n$ and $v_m$ are connected by an edge if and only if there is a tunable admittance interconnecting the $n$th and $m$th RIS elements.

\section{Planar-Connected RIS}

In this section, we exploit the graph theoretical model of BD-RIS to identify the BD-RIS architectures that are implementable in a double-layer \gls{pcb}, where one layer is occupied by the interconnections and the other is the ground.
Such architectures are practically useful since having a \gls{pcb} with more layers increases the implementation complexity.
We begin by observing that a BD-RIS can be implemented into a double-layer \gls{pcb} when its interconnections can be routed on a single layer without crossings.
This means that its associated graph is a \textit{planar graph}, i.e., it can be drawn on the plane such that no edges cross each other.
To better clarify the graph theoretical definition of planar graph \cite[Chapter~9]{bon76}, we report two examples of non-planar and planar graphs in Fig.~\ref{fig:graphs}.
Following this observation, we refer to BD-RIS architectures that are implementable in a double-layer \gls{pcb} as planar-connected RISs, formally defined as follows.
\begin{definition}
(Planar-connected RIS)
A BD-RIS architecture with associated graph $\mathcal{G}$ is denoted as planar-connected when $\mathcal{G}$ is a planar graph for any number of RIS elements $N$.
\end{definition}
Note that every RIS element also needs to be connected to ground through an admittance component, which is not captured in $\mathcal{G}$.
Nevertheless, the ground layer can be reached from every point of the board through vias, always without crossing any interconnection.

\begin{figure}[t]
\centering
\includegraphics[width=0.28\textwidth]{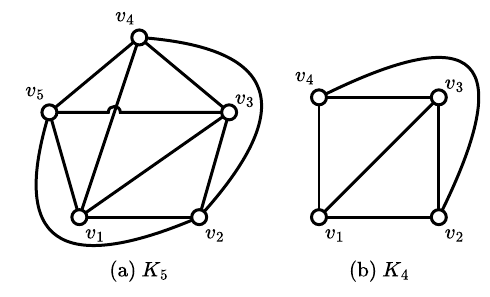}
\caption{(a) The complete graph on five vertices $K_5$ (non-planar), and (b) the complete graph on four vertices $K_4$ (planar).}
\label{fig:graphs}
\end{figure}

After determining the requirements on the graph $\mathcal{G}$ for a BD-RIS to be planar-connected, we study in the following four propositions when existing BD-RIS architectures are planar, i.e., whether they are planar-connected or not, and under what conditions.

\begin{proposition}
Any forest-connected RIS (\cite{ner24}), including the single-connected RIS and the tree-connected RIS as special cases, is planar-connected.
\label{pro:forest}
\end{proposition}
\begin{proof}
To prove that any forest-connected RIS is planar, we recall that the graph $\mathcal{G}$ of a forest-connected RIS is \textit{acyclic}, i.e., it does not contain any cycle (a finite sequence of distinct edges joining a sequence of vertices where only the first and last vertices are equal) \cite{ner24}.
Since the graph of a forest-connected RIS does not have any cycle, it can always be drawn on the plane with no crossings, and therefore is planar.
\end{proof}

\begin{proposition}
A group-connected RIS (\cite{she22}) is planar-connected if and only if the group size is $N_G\leq4$, and the fully-connected RIS (\cite{she22}) is not planar-connected.
\label{pro:group}
\end{proposition}
\begin{proof}
To verify whether a group-connected RIS is planar, we recall a group-connected RIS with group size $N_G$ has an associated graph whose components are the complete graphs on $N_G$ vertices, denoted as $K_{N_G}$.
Thus, a group-connected RIS with group size $N_G$ is planar if and only if $K_{N_G}$ is a planar graph, which can be verified with Kuratowski's Theorem.
This theorem states that a graph is planar if and only if it does not contain any subdivision of $K_{5}$ (the complete graph on five vertices) or $K_{3,3}$ (the complete bipartite graph on six vertices, where three vertices are connected to all other three) \cite[Theorem~9.10]{bon76}.
Therefore, a group-connected RIS with group size $N_G$ is planar if and only if $N_G\leq4$.
Otherwise, $K_{5}$ is contained in the graph $\mathcal{G}$ of the RIS, which becomes non-planar.

The fully-connected RIS is not planar since the definition of planar-connected RIS (see Definition~1) requires its graph to be planar for any number of RIS elements $N$, which is not satisfied for $N\geq5$.
\end{proof}

\begin{proposition}
A $Q$-stem-connected RIS (\cite{zho25a}) is planar-connected if and only if $Q\leq2$.
\label{pro:stem}
\end{proposition}
\begin{proof}
First, for $Q\leq2$, the $Q$-stem-connected RIS has an associated graph $\mathcal{G}$ where only one or two vertices (known as central vertices) are connected to all others \cite{zho25a}.
Since such a graph $\mathcal{G}$ cannot contain subdivisions of the complete graph $K_{5}$ or the complete bipartite graph $K_{3,3}$, $\mathcal{G}$ is planar \cite[Theorem~9.10]{bon76}.

Second, any $Q$-stem-connected RIS with $Q\geq3$ and $N\geq6$ has an associated graph $\mathcal{G}$ that contains $K_{3,3}$-type subgraphs induced by the three central vertices and any other three vertices, and thus is not planar \cite[Theorem~9.10]{bon76}.
\end{proof}

\begin{proposition}
A $Q$-band-connected RIS (\cite{wu25}) is planar-connected if and only if $Q\leq3$.
\label{pro:band}
\end{proposition}
\begin{proof}
First, we prove that a $Q$-band-connected RIS with $Q\leq3$ is planar by proving that the $3$-band-connected RIS is planar.
This is sufficient since the graphs of $1$- and $2$-band-connected RISs are subgraphs of the graph of the $3$-band-connected RIS, and therefore are planar if the graph of the $3$-band-connected RIS is planar.
We prove that the $3$-band-connected RIS is planar by recursively constructing a planar drawing of its graph for any number of elements $N$.
Assume that the graph with $N-1$ vertices has a planar drawing.
To construct the graph for $N$ vertices, we add a new vertex $v_{N}$ and connect it to the last three existing vertices, namely $v_{N-3}$, $v_{N-2}$, and $v_{N-1}$.
For this recursive construction to preserve planarity, the vertices $v_{N-3}$, $v_{N-2}$, and $v_{N-1}$ must all lie on the same face of the current drawing, e.g., the outer face.
If this condition is satisfied, we can always insert the new vertex $v_{N}$ within that face and draw its three edges without introducing any crossings and maintaining the vertices $v_{N-2}$, $v_{N-1}$, and $v_{N}$ on the outer face, allowing the process to continue recursively.
This process is graphically shown in Fig.~\ref{fig:example1}, where we start from the case $N=3$ by drawing the graph as a triangle with vertices $v_{1}$, $v_{2}$, and $v_{3}$.
The graph for $N=4$ can be constructed as in Fig.~\ref{fig:example1}(a), by adding $v_{4}$ and its edges.
Then, the graph for $N=5$ can be constructed as in Fig.~\ref{fig:example1}(b), by adding $v_{5}$ and its edges, and so forth.
In Fig.~\ref{fig:example1}(c), we report a planar drawing for the $3$-band-connected RIS with $N=9$ elements constructed with this process.

Second, any $Q$-band-connected RIS with $Q\geq4$ and $N\geq5$ has an associated graph $\mathcal{G}$ that contains the complete graph $K_{5}$, e.g., the subgraph induced by the first five vertices, and thus is not planar \cite[Theorem~9.10]{bon76}.
\end{proof}

\begin{figure}[t]
\centering
\includegraphics[width=0.48\textwidth]{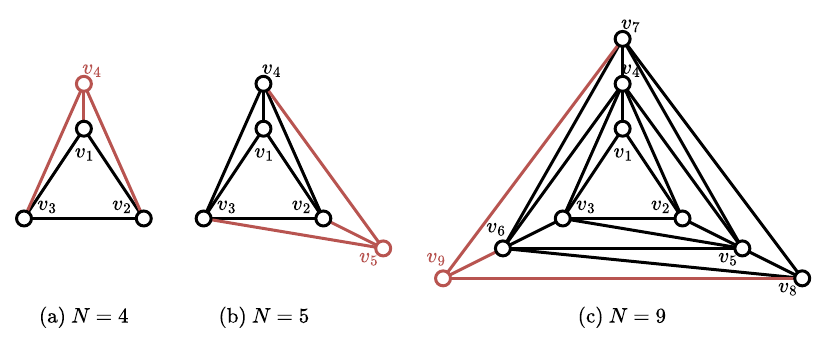}
\caption{Recursive construction of a planar drawing of the graph of a $3$-band-connected RIS with $N$ elements.}
\label{fig:example1}
\end{figure}

\begin{table}[t]
\centering
\caption{Properties of existing BD-RIS architectures.}
\label{tab:planar}
\begin{tabular}{@{}ll@{}}
\toprule
BD-RIS Architecture     & Planar-connected? \\
\midrule
Single-connected RIS [1]    & Planar-connected \\
Group-connected RIS [4]     & Planar-connected iff $N_G\leq 4$ \\
Fully-connected RIS [4]     & Not planar-connected\\
Forest-connected RIS [7]    & Planar-connected \\
Tree-connected RIS [7]      & Planar-connected \\
$Q$-stem-connected RIS [10] & Planar-connected iff $Q\leq2$\\
$Q$-band-connected RIS [12] & Planar-connected iff $Q\leq3$\\
\bottomrule
\end{tabular}
\end{table}

We summarize the findings of Propositions~\ref{pro:forest} to \ref{pro:band} in Table~\ref{tab:planar}, where we report whether the existing BD-RIS architectures are planar-connected.
Single-, forest-, and tree-connected RISs are always planar, while group-, $Q$-stem-, and $Q$-band-connected RISs are planar only under conditions on the group size $N_G$ and $Q$, and the fully-connected RIS is not planar.

\begin{figure}[t]
\centering
\subfigure[Example 1]{
\includegraphics[width=0.12\textwidth]{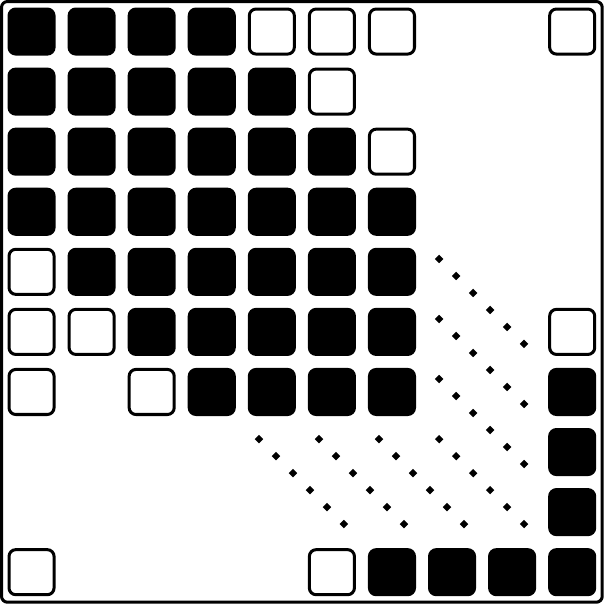}}
\subfigure[Example 2]{
\includegraphics[width=0.12\textwidth]{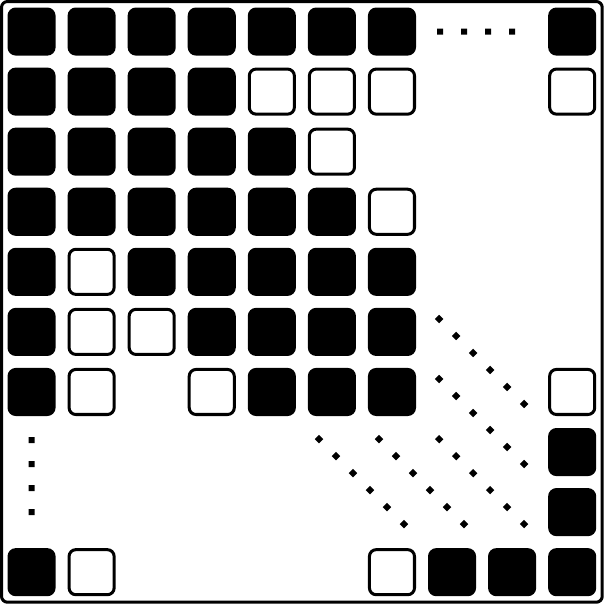}}
\subfigure[Example 3]{
\includegraphics[width=0.12\textwidth]{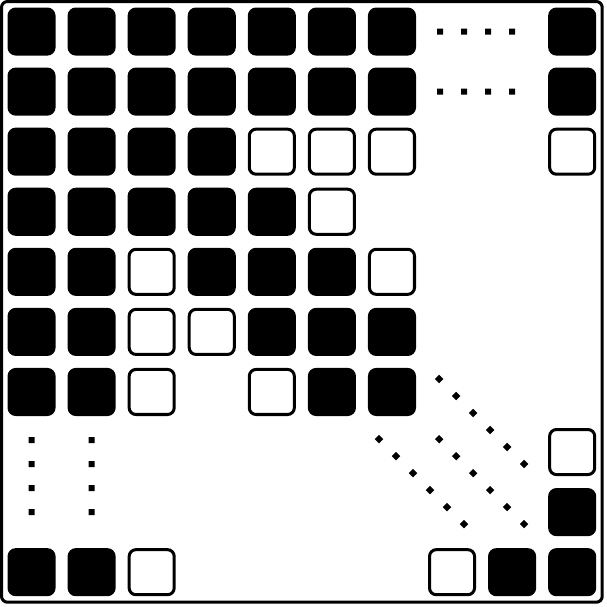}}
\caption{Admittance matrix $\mathbf{Y}$ of the three examples of maximal-planar-connected RIS, with tunable entries in black and zero entries in white.}
\label{fig:y}
\end{figure}

\section{Maximal-Planar-Connected RIS}

We have observed that most of the previously proposed BD-RIS architectures are not planar-connected RISs.
Therefore, in this section, we characterize the most complex planar-connected RIS architectures, i.e., the planar-connected RISs with the largest number of tunable admittance components.

The following proposition provides a necessary condition for a BD-RIS architecture to be planar-connected.
\begin{proposition}
A BD-RIS architecture can be planar-connected only if its graph $\mathcal{G}$ has no more than $3N-6$ edges, with $N\geq3$ denoting the number of RIS elements.
\label{pro:3N-6}
\end{proposition}
\begin{proof}
The proposition is proved since the maximum number of edges in a planar graph on $N$ vertices is $3N-6$ \cite[Theorem~1.4]{fel12}.
\end{proof}
Given Proposition~\ref{pro:3N-6}, we are interested in planar-connected RISs whose interconnection graph has exactly $3N-6$ edges, i.e., known as \textit{maximal planar} graphs in graph theory \cite{fel12}.
We denote this family of BD-RIS architectures as maximal-planar-connected RISs, as formalized in the following definition.
\begin{definition}
(Maximal-planar-connected RIS)
A BD-RIS architecture with associated graph $\mathcal{G}$ is denoted as maximal-planar-connected when $\mathcal{G}$ is maximal planar for any number of RIS elements $N$.
\end{definition}
Maximal-planar-connected RISs are therefore the planar-connected RISs with the highest circuit complexity, given by the number of tunable admittance components.
They include $3N-6$ admittance components interconnecting the RIS elements to each other, and $N$ additional admittance components connecting each RIS element to ground, yielding a total of $4N-6$ admittance components.
Since there are numerous maximal-planar-connected RISs available, we provide three concrete examples.
While these architectures all have the same number of tunable components, their practical implementation may differ in difficulty depending on their specific topology.

\subsubsection{Example 1 ($3$-band-connected RIS)}

The first example of maximal-planar-connected RIS is the $Q$-band-connected RIS with $Q=3$, obtained by interconnecting elements $n$ and $m$ through an admittance if $\vert n-m\vert\leq3$, $\forall n\neq m$, for $n,m=1,\ldots,N$.
The resulting admittance matrix is a banded matrix, specifically heptadiagonal (see Fig.~\ref{fig:y}(a)).
Note that this RIS architecture is planar because of Proposition~\ref{pro:band}, and has $3N-6$ interconnections.
Therefore, it is maximal-planar-connected.

\subsubsection{Example 2}

A second example of maximal-planar-connected RIS can be obtained by interconnecting one element (named the central element) to all other elements through an admittance and by additionally interconnecting elements $n$ and $m$ if $\vert n-m\vert\leq2$, $\forall n\neq m$, for $n,m=1,\ldots,N$.
Following this definition and \eqref{eq:Yij}, the admittance matrix $\mathbf{Y}$ of such an RIS allows non-zero off-diagonal entries only on the two upper/lower diagonals closest to the main diagonal and the first row/column, assuming the central element to be the RIS element $1$ (see Fig.~\ref{fig:y}(b)).
We provide an example of such an RIS having $N=8$ elements in Fig.~\ref{fig:example2}, where we omit the admittances connecting each RIS element to ground.

\begin{figure}[t]
\centering
\includegraphics[width=0.42\textwidth]{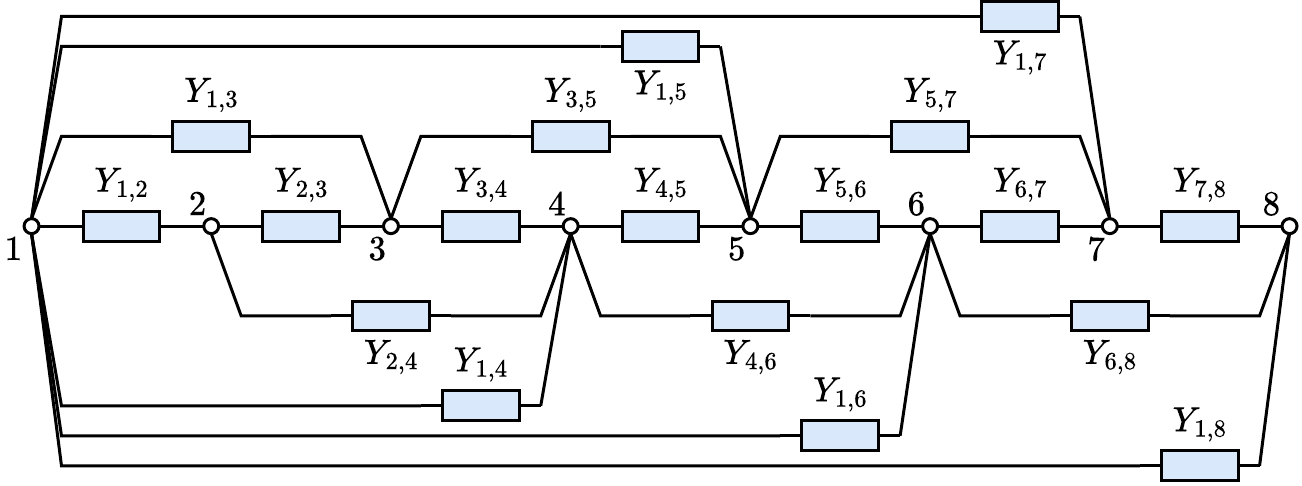}
\caption{Example 2 of maximal-planar-connected RIS, with $N=8$ elements.}
\label{fig:example2}
\end{figure}

\subsubsection{Example 3}

A third example of maximal-planar-connected RIS can be obtained by interconnecting two elements (the central elements) to all other elements through an admittance and by additionally interconnecting elements $n$ and $m$ if $\vert n-m\vert=1$, $\forall n\neq m$, for $n,m=1,\ldots,N$.
The admittance matrix $\mathbf{Y}$ of this RIS allows non-zero off-diagonal entries only on the upper/lower diagonal and the first two rows/columns, assuming the central elements to be the RIS elements $1$ and $2$ (see Fig.~\ref{fig:y}(c)).
We provide an example of such an RIS architecture having $N=8$ elements in Fig.~\ref{fig:example3}.

Remarkably, these three examples of maximal-planar-connected RIS are all optimal architectures in a wireless system with $D=2$ degrees of freedom, i.e., they fulfill the condition in \cite[Theorem~1]{wu25} with $D=2$.
Examples~2 and 3 bridge between the $3$-band- and $3$-stem-connected RIS.

\section{Performance and Complexity Evaluation}

In this section, we compare maximal-planar-connected RISs with the other existing BD-RIS architectures in terms of achieved performance and circuit complexity.
The performance is measured by the sum rate obtained by optimizing the BD-RIS in multi-user \gls{miso} systems, while the circuit complexity is given by the number of tunable admittance components included in the BD-RIS.
While maximal-planar-connected RISs are the planar RISs with maximum flexibility, they are not claimed to achieve optimal sum rate.


Consider an $M$-antenna transmitter serving $K$ single-antenna receivers through the support of an $N$-element BD-RIS.
We denote the channel from the transmitter to the RIS as $\mathbf{H}_{IT}\in\mathbb{C}^{N\times M}$, and from the RIS to the $k$th receiver as $\mathbf{h}_{RI,k}\in\mathbb{C}^{1\times N}$, for $k=1,\ldots,K$.
The channel seen by the $k$th receiver $\mathbf{h}_k\in\mathbb{C}^{1\times M}$ is therefore given by $\mathbf{h}_k=\mathbf{h}_{RI,k}\boldsymbol{\Theta}\mathbf{H}_{IT}$, assuming that the direct links between transmitter and receivers are obstructed (the direct links are neglected to isolate the contribution of the RIS-aided link).
The transmitted signal $\mathbf{x}\in\mathbb{C}^{M\times 1}$ is $\mathbf{x}=\sum_{k=1}^K\mathbf{w}_ks_k$, where $\mathbf{w}_k\in\mathbb{C}^{M\times 1}$ and $s_k\in\mathbb{C}$ are the precoding vector and data symbol for the $k$th receiver.
The precoding vectors are subject to the transmit power constraint $\Vert\mathbf{W}\Vert_F^2\leq P_T$, where $\mathbf{W}=[\mathbf{w}_1,\ldots,\mathbf{w}_K]$, and the data symbols are subject to $\mathbb{E}[\vert s_k\vert^{2}]=1$.
At the $k$th receiver, the received signal $y_k\in\mathbb{C}$ is $y_k=\mathbf{h}_k\mathbf{w}_ks_k+\sum_{i\neq k}\mathbf{h}_k\mathbf{w}_is_i+n_k$, where $n_k\sim\mathcal{CN}(0,\sigma^2)$ is the \gls{awgn}, and the the \gls{sinr} is $\gamma_k=\vert\mathbf{h}_k\mathbf{w}_k\vert^2/(\sum_{i\neq k}\vert\mathbf{h}_k\mathbf{w}_i\vert^2+\sigma^2)$.

\begin{figure}[t]
\centering
\includegraphics[width=0.42\textwidth]{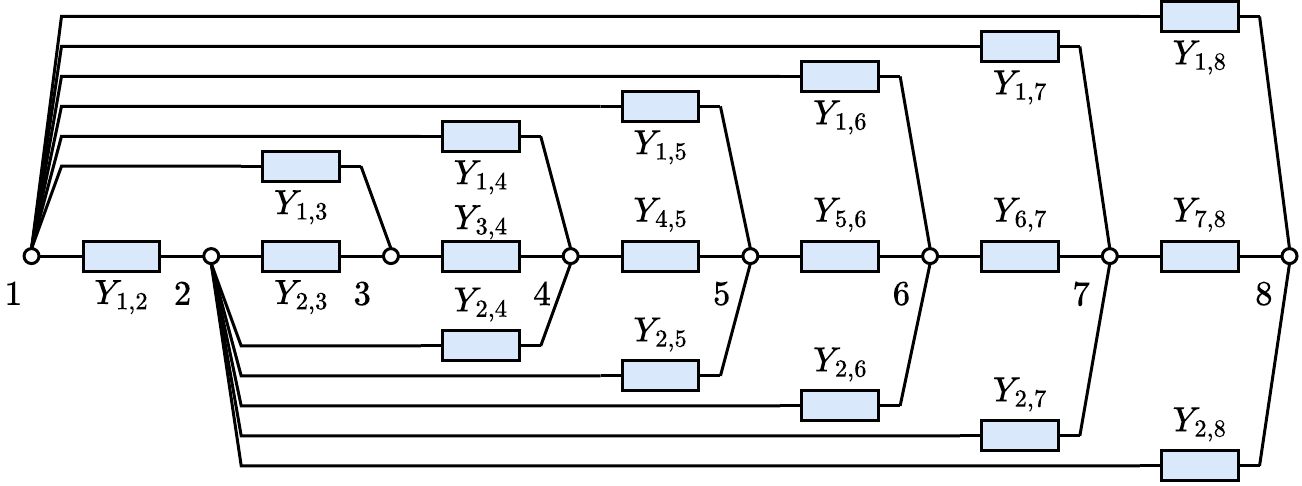}
\caption{Example 3 of maximal-planar-connected RIS, with $N=8$ elements.}
\label{fig:example3}
\end{figure}

We assume the BD-RIS to be made of lossless and reciprocal admittance components, such that the admittance matrix $\mathbf{Y}$ is purely imaginary and writes as $\mathbf{Y}=j\mathbf{B}$, where $\mathbf{B}\in\mathbb{R}^{N\times N}$ is the BD-RIS susceptance matrix, subject to $\mathbf{B}=\mathbf{B}^{T}$.
In addition, for a BD-RIS with associated graph $\mathcal{G}=(\mathcal{V},\mathcal{E})$, the susceptance matrix has specific entries forced to zero, i.e., $\mathbf{B}\in\mathcal{B}_{\mathcal{G}}$, where
\begin{equation}
\mathcal{B}_{\mathcal{G}}=\left\{\mathbf{B}\mid\left[\mathbf{B}\right]_{n,m}=0,\;\forall n\neq m,\;
\left(v_n,v_m\right)\notin\mathcal{E}\right\},
\end{equation}
indicating that the $(m,n)$th entry of $\mathbf{B}$ is forced to zero if elements $m$ and $n$ are not interconnected by an admittance.

\begin{figure*}[t]
\centering
\subfigure[]
{\includegraphics[height=0.20\textwidth]{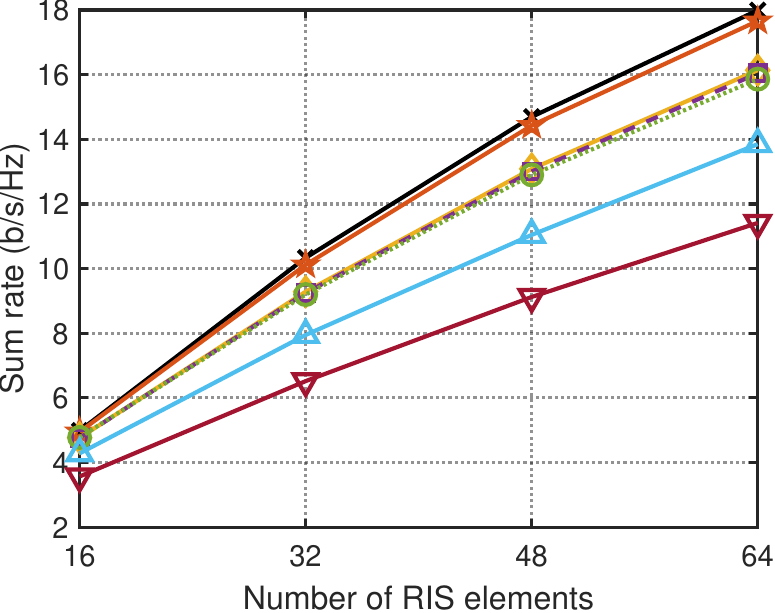}}
\subfigure[]
{\includegraphics[height=0.20\textwidth]{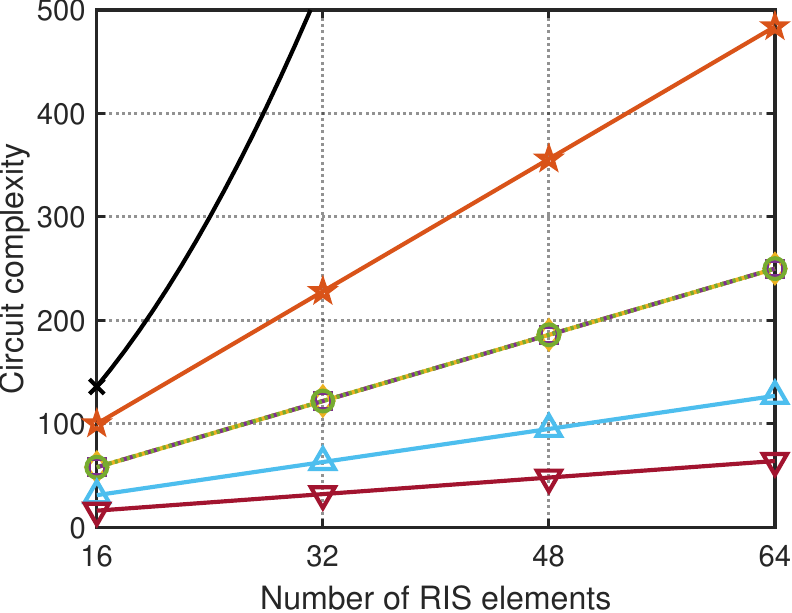}}
\subfigure[$N=64$]
{\includegraphics[height=0.20\textwidth]{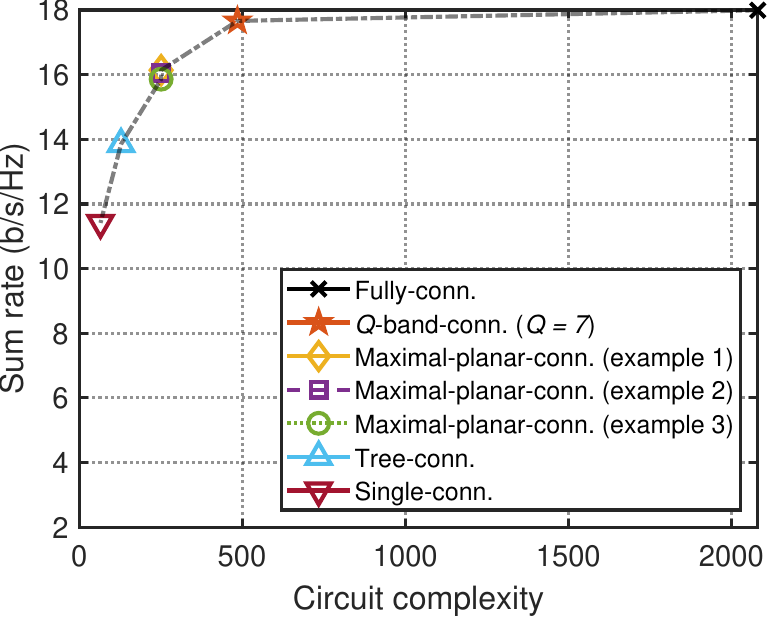}}
\caption{(a) Sum rate, (b) circuit complexity, and (c) their trade-off for different RIS architectures.
The sum rate is calculated with $M=4$ transmit antennas, $K=4$ receivers, $P_T=10$~dBm, and $\sigma^2=-80$~dBm, and averaged over Rayleigh distributed channel realizations.}
\label{fig:results}
\end{figure*}

In this system, the sum rate maximization problem is
\begin{align}
\underset{\mathbf{W},\boldsymbol{\Theta}}{\mathsf{\mathrm{max}}}\;\;
&\sum_{k=1}^K\log_2\left(1+\gamma_k\left(\mathbf{W},\boldsymbol{\Theta}\right)\right)\\
\mathsf{\mathrm{s.t.}}\;\;\;
&\boldsymbol{\Theta}=\left(\mathbf{I}+jZ_{0}\mathbf{B}\right)^{-1}\left(\mathbf{I}-jZ_{0}\mathbf{B}\right),\\
&\mathbf{B}=\mathbf{B}^{T},\;\mathbf{B}\in\mathcal{B}_{\mathcal{G}},\;\left\Vert\mathbf{W}\right\Vert_F^2\leq P_T,
\end{align}
where the precoding matrix $\mathbf{W}$ and the BD-RIS scattering matrix $\boldsymbol{\Theta}$ are the optimization variables.
This optimization problem has been solved in \cite[Section~III]{wu24} through an algorithm valid for any arbitrary BD-RIS architecture, i.e., for any set $\mathcal{B}_{\mathcal{G}}$.
Therefore, we can use that algorithm to also optimize maximal-planar-connected RISs.


To numerically evaluate the sum rate achieved by the different BD-RIS architectures, we consider a multi-user system with $M=4$ transmit antennas, $K=4$ receivers, $P_T=10$~dBm, and $\sigma^2=-80$~dBm.
The channels are generated as \gls{iid} Rayleigh distributed, with path gain given as in \cite[Section~VI]{wu24}.

In Fig.~\ref{fig:results}(a), we report the sum rate achieved with different BD-RIS architectures, including: fully-connected, band-connected with $Q=7$ (proved to achieve the same performance as the fully-connected in this system \cite{wu25}), maximal-planar-connected (examples 1, 2, and 3 whose admittance matrices are represented in Fig.~\ref{fig:y}), tree-connected, and single-connected RISs.
In Fig.~\ref{fig:results}(b), we report the circuit complexity of the same architectures, which is $N(N+1)/2$ for fully-connected \cite{she22}, $(Q+1)(2N-Q)/2=8N-28$ for band-connected \cite{wu25}, $4N-6$ for maximal-planar-connected, $2N-1$ for tree-connected \cite{ner24}, and $N$ for single-connected RISs.
Finally, in Fig.~\ref{fig:results}(c), we report the performance-complexity trade-off enabled by those architectures.
We make the following three observations.
\textit{First}, both fully- and band-connected RISs achieve maximum performance (the tiny gap is due to numerical optimization error), with the band-connected RIS having significantly reduced circuit complexity, which grows linearly with the number of RIS elements.
\textit{Second}, the single-connected RIS achieves the lowest performance since it has limited flexibility, while it also has minimum circuit complexity (it is the planar-connected RIS with minimal complexity, in opposition to maximal-planar-connected RISs).
\textit{Third}, the maximal-planar-connected RISs allow for a favorable balance between performance and circuit complexity, being the most complex BD-RIS architectures implementable in a double-layer \gls{pcb}, i.e., where the interconnections do not intersect.

\section{Conclusion}

In this letter, we have investigated BD-RIS architectures that can be implemented on a double-layer \gls{pcb}, in which one layer is the ground plane.
Using a graph-theoretical modeling, we have identified the necessary conditions for a BD-RIS to be implementable on a double-layer \gls{pcb}, and introduced the concept of planar-connected RIS.
We have further examined existing BD-RIS architectures to determine whether they are planar and characterized the class of maximal-planar-connected RISs, offering the highest flexibility under the double-layer \gls{pcb} constraints.
Our results provide useful guidelines for developing low-complexity BD-RISs.
Future research could investigate the class of planar RISs that maximize performance, and the most flexible RISs implementable with a limited number of \gls{pcb} layers.

\bibliographystyle{IEEEtran}
\bibliography{IEEEabrv,main}
\end{document}